\documentclass[a4paper,USenglish]{lipics-v2018}
\usepackage{inputenc}
\usepackage{amsmath}
\usepackage{amsfonts}
\usepackage{subcaption}
\usepackage{amssymb}
\usepackage{listings}
\usepackage{microtype}%if unwanted, comment out or use option "draft"

\usepackage{csquotes}
\usepackage{pgf}
\usepackage{tikz}
\usepackage{tikz-qtree}
\usetikzlibrary{calc,trees,arrows,graphs,shapes}

\usepackage{xspace}
\newcommand{\BG}{\textsc{BuildMST}\xspace}
\newcommand{\DG}{\textsc{MST}\xspace}

\usepackage[misc]{ifsym}
\usepackage{courier}

\EventAcronym{CVIT}
\ArticleNo{1}
\title{On Underlay-Aware Self-Stabilizing Overlay Networks}

\author{Thorsten Götte$^{(\text{\Letter})}$}{Department of Computer Science, Paderborn University, Paderborn, Germany}{thgoette@mail.upb.de}{https://orcid.org/0000-0001-9798-6993}{}
\author{Christian Scheideler}{Department of Computer Science, Paderborn University, Paderborn, Germany}{scheidel@mail.upb.de}{}{}
\author{Alexander Setzer}{Department of Computer Science, Paderborn University, Paderborn, Germany}{asetzer@mail.upb.de}{}{}
\authorrunning{T. Götte, C. Scheideler and A. Setzer} 

% mandatory: Please provide 1-5 keywords

\subjclass{Theory of computation $\to$ Distributed algorithms; }

\hideLIPIcs
\nolinenumbers
\begin{document}

\maketitle

\begin{abstract}
  We present a self-stabilizing protocol for an overlay network that constructs the Minimum
  Spanning Tree (MST) for an underlay that is modeled by a weighted tree.
    The weight of an overlay edge between two nodes is the weighted length of their shortest path in the tree.
    We rigorously prove that our protocol works correctly under asynchronous and non-FIFO message delivery. Further, the
  protocol stabilizes after $\mathcal{O}(N^2)$ asynchronous rounds where $N$ is the number of nodes in the overlay.
\end{abstract}

\section{Introduction}
The Internet is perhaps the world's most popular medium to exchange any kind of information.
Common examples are streaming platforms, file sharing services or social media
networks. Such applications are often maintained by overlay networks, called overlays for short.
An overlay is a computer network that is built atop another network, the so-called underlay.
In an overlay, nodes that may not be directly connected in the underlay
can create virtual links and exchange messages if they know each others' addresses.
The resulting links then represent a path in the underlying network,
perhaps through several links.

With increasing size of the network, there are several obstacles in designing these overlays.
First of all, errors such as node or link failures are inevitable.
Thus, there is a need for protocols that let the system recover from these faults.
This can be achieved through self-stabilization, which describes a system's ability to reach a desired state from \textsl{any} initial configuration.
Since its conception by Edsger W. Dijkstra in 1975, self-stabilization has proven to be a suitable paradigm to build resilient and scalable overlays
that can quickly recover from changes.
There is a plethora of self-stabilizing protocols for the formation and maintenance of overlay networks with a specific topology.
These topologies range from simple structures like line graphs and rings \cite{onus2007linearization} to more complex overlay networks with useful properties for distributed systems \cite{Feldmann,Richa2011,corona,Skip+}.
These overlays usually minimize the diameter while also maintaining a small node degree, usually at most logarithmic in the number of nodes.
However, the aforementioned overlay protocols are often not concerned with path lengths in the underlay.
This is remarkable, since for many use cases these path lengths
and the resulting latency are arguably more important than the diameter.

In this paper, we work towards closing this gap by proposing a self-stabilizing protocol that
forms and maintains an overlay that resembles the Minimum Spanning Tree (MST)
implied by the distances between nodes in the underlying network.
In particular, we model these distances as a tree metric,
i.e., as the length of the unique shortest path between two overlay nodes in a weighted
tree.
We chose this type of metric because one can find weighted trees in many areas of networking.
In the simplest case, the physical network that interconnects the overlay nodes resembles a tree.
This is often the case in data centers.
Here, the servers are the tree's leaves
while the switches
are the tree's intermediate nodes (cf. \cite{Arregoces03,Leiserson85,Al-Fares08}).
Therefore, we can define a tree metric directly on the paths in this physical infrastructure.
Of course, not all physical networks are strictly structured like trees, and may instead contain cycles.
However, for small networks there are practical protocols that explicitly
reduce the network graph to a tree for routing purposes \cite{ieee_802_1d,Perlman85}.
These protocols are executed directly on the network appliances and exclude certain physical
connections, such that the remaining connections form a spanning tree.
Thus, we can define a tree metric based this tree.
Last, in large-scale networks like the internet neither the physical network nor the routing paths
strictly resemble trees.
However, there is strong evidence that even these large-scale networks can be closely approximated by or embedded into weighted trees
by assigning them virtual coordinates (cf. \cite{Adcock13,Abu-Ata14,Montgolfier11,Shavitt08}).
Thus, we can define a tree metric based on the shortest paths in such an embedding.
In summary, tree metrics promise to be a versatile abstraction
for many kinds of real-world networks.

%The remainder of this paper is structured as follows: In Section \ref{sec:related-work} we present related work.
%Afterwards, we introduce some preliminaries in Section \ref{sec:preliminaries}.
%As the main part of this paper, Sections \ref{sec:protocol} and \ref{sec:analysis} present and rigorously analyze our protocol \BG.

\subsection{Model \& Definitions}
\label{sec:model}
We consider a distributed system based on a fixed set of nodes $V$.
Each node $v \in V$ represents a computational unit, e.g., a computer, that possesses
a set of local variables and references to other nodes, e.g, their IP addresses.
These references are immutable and cannot be corrupted.
If clear from the context,
we refer to the reference of some node $w \in V$ simply as $w$.
Further, each node in $V$ has access to a \emph{tree metric} $d_T: V^2 \rightarrow \mathbb{R}^+$
that assigns a weight to each possible edge in the overlay.
In particular, the function $d_T$ returns the weighted length of the unique shortest path between two nodes in the weighted
tree $T := (V_T,E_T,f)$ with $f: V^2 \rightarrow \mathbb{R}^+$ and $V_T \supseteq V$.
A node $v \in V$ can check the distance $d_T(v,w)$
only if it has a reference to $w \in V$ in its local variables.
Furthermore, it can check the distance $d_T(u,w)$
of all nodes $u,w \in V$ in its local variables.
Throughout this paper, we refer to the metric space $(V,d_T)$
also as a tree metric for ease of description.
%For ease of description we will refer to $d_T(v,w)$
%as the length of $\{v,w\}$.

Sending a message from a node $u$ to another node $v$ in the overlay is only possible if $u$ has a reference to $v$.
All messages for a single node are stored in its so-called channel and we assume \emph{fair message receipt},
which means each message will \emph{eventually} be received.
In particular, we do \emph{not} assume FIFO-delivery, i.e., the messages may be received and processed in any order.

We assume that each node runs a protocol that can perform computations
on the node's local variables and send them within messages to other nodes.
To formalize the protocol's execution, we use the notion of \textsl{configurations}.
A configuration $c$ contains each node's internal state, i.e., its assignment of
values to its local variables, its stored references, and all messages in the node's channel.
We denote $C$ to be the set of all possible configurations.
Further, a \emph{computation} is an infinite series of configurations $(c_t, c_{t+1}, \dots)$,
such that $c_{i+1}$ is a succeeding configuration of $c_{i}$ for $i \geq t$ according to the protocol.
In each step from $c_i$ to $c_{i+1}$, the following happens:
One node $v \in V$ is activated and an arbitrary (possibly empty) set of messages from $v$'s channel is delivered to $v$.
Once activated, the node will execute its protocol and processes all messages delivered to it.
As we do not specify which node is activated and which messages get delivered,
there are maybe several possible succeeding configurations $c' \in C$ for any configuration $c$.
Last, we assume \emph{weakly fair execution}, which means that each node is eventually activated.
Other than that, we place no restriction on the activation order.

Given a subset $C' \subseteq C$, we say that the system \textit{reaches} $C'$ from $c_t$
if \textit{every} computation that starts in configuration $c_t$ eventually contains a configuration $c_{t'} \in C'$.
Note that this does not imply that any succeeding configuration of $c_{t'}$ is in $C'$ as well.
%Last, we define the set $C_{t'>t} \subset C$
%that contains all possible configurations $c_{t'} \in C$ that can ever be reached starting from configuration $c_t$.

Based on this notion of configurations, we can now define self-stabilization.
A protocol is self-stabilizing concerning a set of legal configurations $L \subseteq C$
if starting from any initial configuration $c_0 \in C$ each computation will
eventually reach $L$ (Convergence) and every succeeding configuration is also in $L$ (Closure).
Formally:

\begin{definition}[Self-Stabilization]
    \label{def:self-stabilization}
    A protocol $\mathcal{P}$ is self-stabilizing if it fulfills the following two properties.
    \begin{enumerate}
        \item (Convergence) Let $c_0 \in C$ be \textsl{any} configuration.
        Then every computation that starts in $c_0$ will reach $L$ in finitely many steps.
        \item (Closure) Let $c_t \in L$ be \textsl{any} legal configuration.
        Then \textsl{every} succeeding configuration of $c_t$ is legal as well.
    \end{enumerate}
\end{definition}

Throughout this paper we distinguish between two kinds of edges in each configuration $c \in C$.
We call an edge $(v,w) \in V^2$ \emph{explicit} if and only if $v$ has a reference to $w$ stored in its local variables.
Otherwise, if the reference is in $v$'s channel, we call the edge \emph{implicit}.
Based on this definitions,
we define the directed graph $G_c := (V, E_c^X \cup E^T_c)$ where the set $E^X_c \subseteq V^2$ denotes the set of explicit edges and
$E_c^T \subseteq V^2$ denotes the set of implicit edges.
Further, the undirected graph $G_c^* := (V,E^*_c)$ arises from $G_c$
if we ignore all edges' direction and whether they are implicit or explicit.

%Let $c^i_k \in C$ the configuration the first configuration where a) all references in $E^I_{c^i_0}$ have been delivered to their destination and b) all destination have been activated once after that.
%Then all configurations $c^i_0, \dots, c^i_k$ up until this point are part of round $\mathcal{R}_i$.
%Last, let $c'$ be a succeeding configuration of $c^i_k$, then $c:= c_0^{i+1}$ is the first configuration of round $R_{i+1}$.

%Note that we cannot simply use the number of configurations as measure for convergence time.
%Recall that in an asynchronous setting the nodes' actions can be executed in almost any order as long as their guard is fulfilled.
%We assumed that a third party called d\ae mon decides which actions are executed and which messages are delivered.
%If we only count configurations, a malicious d\ae mon could create an arbitrarily long (but finite) convergence time by withholding crucial messages.

%The strategy could be as follows: All nodes can perform the action guarded by the trivial guard \textsl{True} in every configuration.
%But only if certain nodes are activated, the system moves closer to convergence.
%Therefore, nodes which do not improve the system could be activated arbitrarily
%often before a node which actually brings the system closer to converegence is activated.

\subsection{Our Contribution}
\label{sec:protocol_tree}

Our main contribution is \BG,
a self-stabilizing protocol that forms and maintains overlay representing the MSTs
of all connected components of $V$.
An MST is a set of edges that connects a set of nodes
and minimizes the sum of the edges' weights given by the
underlying metric.
Because of this minimality,
it can serve as a building block for more elaborate topologies.
Note that in our model it is not always possible to construct the MST of all nodes, even if it is unique.
To exemplify this, consider an initial configuration $c_0$ where $G^*_{c_0}$ is \emph{not} connected.
Then two nodes from two different connected components of $G^*_{c_0}$
can never communicate with each other and create edges because they cannot learn each other's reference.
This was remarked in \cite{corona}.
In this case, it is impossible to construct an MST for all nodes as no protocol can add the necessary edges.
Instead one can only construct the MST of all initially connected components, i.e., a Minimum Spanning Forest.

Formally an MST is defined as follows.
\begin{definition}[Minimum Spanning Tree]
    Let $G := (V,E)$ be a graph and $f: E \to \mathbb{R}^+$ a weight function, then the Minimum Spanning Tree $MST(G,f) \subseteq E$ is a set of edges, such that:
    \begin{enumerate}
        \item $(V,MST(G,f))$ is a connected graph, and
        \item $\sum_{e \in MST(G,f)} f(e)$ is minimum.
    \end{enumerate}
    For the special case of $E := V^2$, i.e., the MST over all possible edges, we write $MST(V,f)$ for short.
\end{definition}

In this paper, we will only consider metrics with distinct distances
for each pair of nodes.
Otherwise the MST may not be unique for a metric space $(V,d_T)$.
If we had edges with equal distances, we would need to employ some mechanism of tie-breaking,
e.g., via the nodes' identifiers.

%%%%%%%%%%%%%%%%%%%%%%%%%%%%%%%%%%%%%%%%%%%%%%%%%%%%%%%%%%%%%%%%%%%%%%%%%%%%%%%%
%% Legale Konfigurationen
%%%%%%%%%%%%%%%%%%%%%%%%%%%%%%%%%%%%%%%%%%%%%%%%%%%%%%%%%%%%%%%%%%%%%%%%%%%%%%%%
In the following, we define the set $\mathcal{L}_{MST} \subset C$ of legal configurations for \BG.
We regard all configurations $c \in \mathcal{L}_{MST}$ as legal in which the explicit edges form the MST of each connected component in $c$.
Further, a legal configuration may contain arbitrarily many implicit edges as long as they are part of an MST.
Formally:

\begin{definition}[Legal Configurations $\mathcal{L}_{MST}$]
    \label{def:legal_rng}
    Let $(V,d_T)$ be a tree metric and $c \in C$ be a configuration.
    Further denote $G_1, \dots, G_k$ as the connected components of $G^*_{c}$.
    Then the set of legal configurations $\mathcal{L}_{MST}$ is defined by the following two conditions:
    \begin{enumerate}
        \item A configuration $c \in \mathcal{L}_{MST}$ contains an explicit edge $(v,w) \in E^X_{c} $ if and only if there is component $G_i:=(V_i,E_i)$ with $\{v,w\} \in MST(V_i,d_T)$.
        \item A configuration $c \in \mathcal{L}_{MST}$ contains an implicit edge $(v,w) \in E^T_{c} $ only if there is a component $G_i:=(V_i,E_i)$ with $\{v,w\} \in MST(V_i,d_T)$.
    \end{enumerate}
\end{definition}

%\subsection{Organization of this Paper}

%%%%%%%%%%%%%%%%%%%%%%%%%%%%%%%%%%%%%%%%%%%%%%%%%%%%%%%%%%%%%%%%%%%%%%%%%%%%%%%%

%%%%%%%%%%%%%%%%%%%%%%%%%%%%%%%%%%%%%%%%%%%%%%%%%%%%%%%%%%%%%%%%%%%%%%%%%%%%%%%%

\section{Related Work}
\label{sec:related-work}

There are several self-stabilizing protocols for constructing spanning trees in a fixed
communication graph, e.g., \cite{Blin09,Gheorge97,BlinSSS10,Higham01,Blin10,Korman}.
These works do not consider a model where nodes can create arbitrary overlay edges.
Instead, each node has a \emph{fixed} set of neighbors
and chooses a subset of these neighbors for the tree.
Furthermore, the communication graph in all these works is modeled as an arbitrary weighted graph instead of a tree.
The fastest protocol given in \cite{Blin10} constructs an MST in $\mathcal{O}(N^2)$ rounds where $N$ is the number of nodes.
Note that \cite{Korman} proves the existence of a protocol that converges in $\mathcal{O}(N)$
rounds but does not present and rigorously analyze an actual protocol.
As stated in the introduction, these protocols can be used in the underlying network to construct
a tree metric for our protocol.

In the area of topological self-stabilization of overlay networks,
there is a plethora of works
that consider different topologies like line graphs \cite{onus2007linearization}, De-Bruijn-Graphs \cite{Feldmann,Richa2011}, or Skip-Graphs \cite{Skip+,corona}.
Besides these results that do not take the underlying network into account,
there are also efforts to build a topology based on a given metric.
An interesting result in this area is a protocol for building the Delaunay Triangulation of two-dimensional metric space by Jacob et al. \cite{JACOB2012137}.
This work bears several similarities with ours.
In particular, the Delaunay Triangulation is a superset of the metric's MST and shares some of the properties we present in Section \ref{sec:preliminaries}.
Also their protocol $D_{STAB}$ is very similar to our protocol \BG.
Recently Gmyr et al.\ proposed a self-stabilizing protocol for constructing an overlay based on an arbitrary metric \cite{Gmyr2016}.
Instead of building a spanning tree, their goal is to build an overlay in which the distance between two nodes is exactly the distance in the underlying metric.
In particular, their algorithm is also applicable to a tree metric.
However, note that for a tree metric the number of edges in the resulting overlay
can be as high as $\Theta(N^2)$.

Last, there are several non-self-stabilizing approaches for creating underlay-aware overlays, e.g., \cite{Gross12,Rowstron2001,abraham2004land,plaxton1999}.
With their often-cited work in \cite{plaxton1999}, Plaxton et al.\
introduced these so-called location-aware overlays.
The authors present an overlay for an underlay modeled by a growth-bounded two-dimensional metric.
This means that the number of nodes within a fixed distance of a node
only grows by a factor of $\Delta \in \mathbb{R}^+$ when doubling the distance.
Their overlay has a polylogarithmic degree and the length of the routing paths
approximate the distances in the underlying metric by a polylogarithmic factor.
In \cite{abraham2004land} Abraham et al.\ extended on \cite{plaxton1999} and proposed an overlay for growth-bounded metrics where
the latter is reduced to a factor of $1+\epsilon$. Here, $\epsilon \in \mathbb{R}^+$ is a parameter that can be set to an arbitrarily small value.
The resulting overlay's degree depends on $\epsilon$ and is not analyzed in detail.

\section{Preliminaries}
\label{sec:preliminaries}
In this section, we present some useful properties of tree metrics and their MSTs
that will help us in designing and analyzing our protocol.
Therefore, we introduce the notion of \emph{relative neighbors}.
Two nodes $v,w \in V$ are relative neighbors with regard to a metric $d_T$ if there is no third node
that is closer to either of them,
i.e., it holds $\nexists u \in V: \left(d_T(u,v) < d_T(v,w)\right) \wedge \left(d_T(u,w) < d_T(v,w) \right)$.
Throughout this paper we write $u \prec (v,w)$ as shorthand for $\left(d_T(u,v) < d_T(v,w)\right) \wedge \left(d_T(u,w) < d_T(v,w)\right)$.
Relative neighbors have been defined and analyzed for a variety of metrics (cf. \cite{Toussaint80,Jaromczyk92,Supowit83}), but they
prove to be especially useful in the context of tree metrics.
In particular, they allow nodes to form and maintain an MST based on local criteria.
This fact is stated by the following lemma:
\begin{lemma}
\label{lemma:equivalence}
Let $(V,d_T)$ be a tree metric, then the following two statements hold:
\begin{enumerate}
    \item  $\{v,w\} \in MST(V,d_T) \Longrightarrow \nexists u \in V: \, u \prec (v,w)$
    \item  $\{v,w\} \not\in MST(V,d_T) \Longrightarrow \exists u \in V: \, \big( u \prec (v,w) \wedge \{v,u\} \in MST(V,d_T) \big)$
\end{enumerate}
\end{lemma}
In the following, we will outline the proof and thereby present some helpful lemmas,
which we will reuse in Section \ref{sec:analysis}.
First, we note that the lemma's first statement is generally true for all metrics (cf. \cite{Supowit83}).
Thus, it remains to show the second statement.
We begin the proof with a useful fact that will be at the core of many proofs in this paper.
\begin{lemma}
\label{thm:root}
Let $(V,d_T)$ be a tree metric.
Further let $u,v,w,r \in V$ be four nodes, s.t.
\[
    d_T(u,r) < d_T(w,r) \,\wedge\, d_T(v,r) < d_T(w,r)
\]
Then it either holds $u \prec (v,w)$
or $v \prec (u,w)$ (and in particular not $w \prec (u,v)$).
\end{lemma}

\begin{proof}
Let $T:=(V_T,E_T,w)$ be the tree which implies the metric $d_T$.
Further denote the unique path shortest between two nodes $s,t \in V_T$ in $T$ as $P_T(s,t)$.
Last, let $\varphi \in P_T(u,v) \cap P_T(u,r) \cap P_T(v,r) $ be a node that lies on all three unique shortest
paths between the nodes $u,v$ and $r$.
Note that $\varphi$ is also called \emph{median} of $u,v$ and $r$ and is unique in a tree.
First, we show that
\[
        d_T(u,\varphi) < d_T(w,\varphi) \,\,\, \wedge \,\,\, d_T(v,\varphi) < d_T(w, \varphi)
\]
Assume for contradiction that $d_T(u,\varphi) > d_T(w,\varphi)$.
From the triangle inequality we can follow that $d_T(w,r) \leq d_T(w,\varphi) + d_T(\varphi,r)$.
If we combine these two inequalities, we deduce
\[
 d_T(w,r) \leq d_T(w,\varphi) + d_T(\varphi,r) < d_T(u,\varphi) + d_T(\varphi,r) = d_T(u,r)
\]
This would be a contradiction to our initial assumption that $d_T(u,r) < d_T(w,r)$.
Therefore, it must hold $d_T(u,\varphi) < d_T(w,\varphi)$.
The proof for $d_T(v,\varphi) < d_T(w, \varphi)$ is analogous and thus, our claim holds.

Second, we prove that it holds $\varphi \in P_T(v,w)$ or $\varphi \in P_T(u,w)$.
Assume for the sake of contradiction that neither $\varphi \in P_T(v,w)$ nor $\varphi \in P_T(u,w)$.
Then there is a path from $u$ to $v$ via $w$ that does not contain $\varphi$.
This is a contradiction to the fact that there is only
one simple path $P_T(u,v)$ between $u$ and $v$
and per definition it holds $\varphi \in P_T(u,v)$.

Now distinguish between the two cases we have just shown:
\begin{enumerate}
    \item If $\varphi \in P(v,w)$, the following inequality must hold.
    \begin{align*}
        d_T(u,v) =  d_T(u,\varphi) + d_T(\varphi,v) < d_T(w,\varphi) + d_T(\varphi,v) = d_T(v,w)
    \end{align*}
    This follows from the fact that $d_T(u,\varphi) < d_T(w,\varphi)$.
    Since $d_T(v,w) < d_T(u,v)$ is one of the two requirements for $w \prec (u,v)$, it cannot hold in this case.
    \item Otherwise, if $\varphi \in P(u,w)$, it holds
    \begin{align*}
        d_T(u,v) =  d_T(u,\varphi) + d_T(\varphi,v) < d_T(u,\varphi) + d_T(\varphi,w) = d_T(u,w)
    \end{align*}
    This follows from the fact that $d_T(v,\varphi) < d_T(w,\varphi)$.
    Since $d_T(u,w) < d_T(u,v)$ is required for $w \prec (u,v)$, it cannot hold in this case either.
\end{enumerate}
Hence, it must hold $u \prec (v,w)$ or $v \prec (u,w)$, which was to be shown.
\end{proof}

%Otherwise, we could construct a path from $u$ to $w$ via $w$
%that does not contain $\varphi$, which contradicts its definition.
%First, we note that $u$ and $v$ are closer to $\varphi$ than $w$.
%If $w$ were closer to $\varphi$ than $v$ or $w$, it would also be closer to $r$ than the respective node
%since we could constuct a path from $w$ to $r$ via $\varphi$.
%This is a contradiction to our assumption.
%Further, it must hold that either the path from $u$ to $w$ or from $v$ to $w$ must
%contain $\varphi$. Otherwise, we could construct a path from $u$ to $w$ via $w$
%that does not contain $\varphi$, which contradicts its definition.
%W.l.o.g. let $\varphi$ be on the path from $u$ to $w$, then $u$ is closer
%to $v$ than to $w$ because $v$ is closer to $\varphi$.
%Since $d_T(u,w) < d_T(u,v)$ is required for $w \prec (u,v)$, it cannot hold.
%An analogous argument holds for the case that $\varphi$ is on the pathe from $v$ to $w$.

%%%%%%%%%%%%%%%%%%%%%%%%%%%%%%%%%%%%%%%%%%%%%%%%%%%%%%%%%%%%%%%%%%%%%%%%%%%%%%%%
%% Greedy Routing
%%%%%%%%%%%%%%%%%%%%%%%%%%%%%%%%%%%%%%%%%%%%%%%%%%%%%%%%%%%%%%%%%%%%%%%%%%%%%%%%
Using Lemma \ref{thm:root} we can show the following.
\begin{lemma}
    \label{lemma:greedy-routing}
    Let $(V,d_T)$ be a tree metric and $v,w \in V$ two of its nodes.
    Further, let  $v_0, \dots, v_k \in V$ be the unique path from $v_0 := v$ to $v_k := w$ in the MST.
    Then it holds:
    \[
        d_T(v_{i},v) < d_T(v_{i+1},v) \,\,\, \forall v_i \in (v_0, \dots, v_{k-1})
    \]
\end{lemma}
%If we assume there is a path where the property does not hold, there must be a first deviator.
%This is a node $v_i \in P_{MST}(v,w)$ with $d_T(v_{i-1},v) > d_T(v_{i},v)$,
%such that the lemma holds for all its predecessors on the path.
%In particular, it holds $d_T(v_{i-1},v) > d_T(v_{i-2},v)$ for its direct predecessors $v_{i-1}$ and $v_{i-2}$.
%Now we can use Lemma \ref{thm:root} to show that the MST could be improved by
%swapping either $(v_{i-1},v_i)$ or $(v_{i-2},v_{i-1})$ for $(v_{i-2},v_{i})$.
\begin{proof}[Lemma \ref{lemma:greedy-routing}]
    Assume for the sake of contradiction that the lemma does not hold.
    Let $v_i$ be the first node on a path to $v$ for which it instead holds $d_T(v_{i}, v) > d_T(v_{i+1},v)$.
    Note, that it cannot hold $d_T(v_{i}, v) = d_T(v_{i+1},v)$ because we assume pairwise distinct distances.
    Since $d_T(v,v_0) = 0$ it must hold that $i \geq 1$.
    Therefore $v_{i-1}$ is well-defined and it must hold $d_T(v_{i-1}, v) < d_T(v_{i},v)$
    because $v_i$ is the first node that is further away from $v$ than its successor.
    Combining these two facts yields:
    \[
        d_T(v_{i-1},v) < d_T(v_{i},v) \,\,\, \wedge \,\,\, d_T(v_{i+1},v) < d_T(v_{i},v)
    \]
    Following Lemma \ref{thm:root} it must therefore either hold $v_{i-1} \prec (v_{i},v_{i+1})$ or $v_{i+1} \prec (v_{i-1},v_{i})$.
    In particular that means, it holds either hold $d_T(v_{i-1}, v_{i+1}) < d_T(v_{i},v_{i+1})$ or $d_T(v_{i-1}, v_{i+1}) < d_T(v_{i-1},v_{i})$.
    In the following we assume the latter since both cases are analogous.
    We will now show that we can improve the MST by swapping $\{v_{i-1}, v_{i+1}\}$ for $\{v_{i-1},v_{i}\}$, which is a contradiction.
    If we remove $\{v_{i-1},v_{i}\}$ from $MST(V,d_T)$ we divide the tree
    into two subtrees $T_{i-1}$ and $T_i$, which contains $v_{i-1}$ and $v_{i}$, respectively.
    Further, it holds that $v_{i+1}$ is in $T_i$ because it connected to $v_i$ via the edge $\{v_{i},v_{i+1}\}$.
    Thus, the edge $\{v_{i-1}, v_{i+1}\}$ also connects $T_{i-1}$ and $T_i$ and has lower weight
    than $\{v_{i},v_{i-1}\}$. That means, we can improve $MST(V,d_T)$, which is a contradiction.
    Therefore, there cannot be such a first deviator $v_i$ and the lemma must hold.
\end{proof}

In the remainder, we conclude the proof for Lemma \ref{lemma:equivalence}.
Therefore, let $v,w \in V$ be two nodes with $\{v,w\} \not\in MST(V,d_T)$.
Further, let $u \in V$ be the first node of the path $P_{vw}$ from $v$ to $w$ in the MST.
Such a node must exist because there is no direct edge between $v$ and $w$ in the MST.
Note that $P_{vw}$ contains the same nodes as a path $P_{wv}$ from $w$ to $v$ but in reverse order.
Thus, we can apply Lemma \ref{lemma:greedy-routing} in "both directions".
That means, the node $u$ with $\{v,u\} \in MST(V,d_T)$ must be closer to $w$ than $v$,
but also closer to $v$ than its successor in $P_{wv}$.
A simple induction then yields that $u \prec (v,w)$.
Since by definition it holds $\{v,u\} \in MST(V,d_T)$, this proves the lemma.
\section{Protocol}
\label{sec:protocol}

%%%%%%%%%%%%%%%%%%%%%%%%%%%%%%%%%%%%%%%%%%%%%%%%%%%%%%%%%%%%%%%%%%%%%%%%%%%%%%%%
%% Protocol
%%%%%%%%%%%%%%%%%%%%%%%%%%%%%%%%%%%%%%%%%%%%%%%%%%%%%%%%%%%%%%%%%%%%%%%%%%%%%%%%

In this section, we describe our protocol \BG,
which forms and constructs an overlay according to Definition \ref{def:legal_rng}.
Intuitively, the protocol works as follows:
Upon activation, a node $v \in V$ checks, which of its current neighbors are relative
neighbors.
All nodes that fulfill the property are kept in the neighborhood.
All others are delegated in a greedy fashion.
This idea resembles that of the protocols in \cite{JACOB2012137} and \cite{onus2007linearization}, where essentially the same technique is used for different underlying metrics,
i.e., the two-dimensional plane and a line.

\begin{lstlisting}[caption={\BG},label=alg:BuildDG,captionpos=b,float, mathescape=true, belowskip=-\baselineskip, frame = single]
Upon activation a node $v \in V$ performs:
              for all $w \in N_v$
                  if $\exists u \in N_v: \, u \prec (v,w)$
                    $N_u \longleftarrow N_u \cup \{w\}$ #$v$ delegates $w$ to $u$
                    $N_v \longleftarrow N_v \setminus \{w\}$
                  else
                    $N_w \longleftarrow N_w \cup \{v\}$ #$v$ introduces itself
\end{lstlisting}

The pseudocode for this protocol is given in Figure \ref{alg:BuildDG}.
Therein, each node $v \in V$ only maintains a single variable $N_v \subseteq V$.
This is a set that contains all currently stored references to other nodes.
It contains each entry only once and multiple occurrences of the same reference are merged automatically.

With each activation, a node iterates over all nodes in $w \in N_v$ and checks whether to delegate $w$ or to introduce itself.
In this context, a delegation means that $v$ sends a reference of $w$ to $u$ and then deletes the reference to $w$ from $N_v$.
The protocol assures that a node $v$ delegates $w$ to $u$, if and only if it holds $u \prec (v,w)$.
Otherwise $v$ introduces itself to $w$, which means that it sends a reference of itself to $w$.
Note, that the primitives of introduction and delegation preserve the system's connectivity (cf. \cite{corona}).

In the pseudocode introductions and delegations are indicated by statements of the form $N_u \longleftarrow N_u \cup \{w\}$.
This notation is used for convenience. It describes that the executing node $v$ sends a message containing a reference of
$w$ to $u$.
The variable $N_u$ is not directly changed and $w$ is only added in some later configuration
when $u$ is activated and the message is delivered to $u$.
A graphical example of the protocol's computations can be seen in Figure~\ref{fig:protocol}.

\begin{figure}
    \begin{subfigure}[t]{.49\textwidth}
        \centering
        \begin{tikzpicture}[every tree node/.style={draw,circle, minimum size=1.5em},
        level distance=1cm,sibling distance=1.2cm,
        edge from parent path={(\tikzparentnode) -- (\tikzchildnode)}]
        \tikzset{level 1/.style={level distance=0.905cm, sibling distance=0.16cm}}
        \tikzset{level 2/.style={level distance=0.7cm, sibling distance=1.5cm}}
        \tikzset{level 3/.style={level distance=0.6cm, sibling distance=0.32cm}}
        \Tree [.\node[opacity=1] (w) {$w$};
        \edge[opacity=0.8] node[auto=left] {$6$};[.\node[opacity=0.8, dotted] (phi) {};
        \edge[opacity=0.8] node[auto=right] {$1$}; [.\node (v) {$u$}; ]
        \edge[opacity=0.8] node[auto=left] {$5$}; [
        .\node (u) {$v$};
        ]
        ]
        ]
        \path (u) edge[->, line width=1pt, bend left,  color=red] node[auto=left] (t) {$6$} (v);
        \path (u) edge[->, line width=1pt, bend right, color=red] node[auto=right] (t) {$11$} (w);
        \end{tikzpicture}
        \caption{An example configuration:
        $v$ has $u$ and $w$ in its local memory.
        Note that $u$ and $w$ are neighbors of each other in the \DG.}
    \end{subfigure}
    ~
    \begin{subfigure}[t]{.49\textwidth}
        \centering
        \begin{tikzpicture}[ every tree node/.style={draw,circle, minimum size=1.5em},
        level distance=1cm,sibling distance=1.2cm,
        edge from parent path={(\tikzparentnode) -- (\tikzchildnode)}]
        \tikzset{level 1/.style={level distance=0.905cm, sibling distance=0.16cm}}
        \tikzset{level 2/.style={level distance=0.7cm, sibling distance=1.5cm}}
        \tikzset{level 3/.style={level distance=0.6cm, sibling distance=0.32cm}}
        \Tree [.\node[opacity=1] (w) {$w$};
        \edge[opacity=0.8] node[auto=left] {$6$};[.\node[opacity=0.8, dotted] (phi) {};
        \edge[opacity=0.8] node[auto=right] {$1$}; [.\node (v) {$u$}; ]
        \edge[opacity=0.8] node[auto=left] {$5$}; [
        .\node (u) {$v$};
        ]
        ]
        ]
        \path (u) edge[->,line width=1pt, bend left,  color=red] node[auto=left] (t) {$6$} (v);
        \path (v) edge[->, line width=1pt, color=red, dashed] node[auto=left] (t) {$6$} (u);
        \path (v) edge[->, line width=1pt, bend left, color=red, dashed] node[auto=left] (t) {$7$} (w);
        \end{tikzpicture}
        \caption{The succeeding configuration: $v$ has delegated $w$ to $u$ and introduced itself to $u$.}
    \end{subfigure}

    \caption[An example of the execution of \BG]{An example of the protocol's execution.
    The black edges are part of the underlying tree.
     Red edges denote the overlay's edges.
     The dotted edges are implicit, i.e., the references are still the node's channel. 
     Solid edges are explicit,i.e., the references are in the node's memory.
     The numbers denote the edges' weights.}
    \label{fig:protocol}
\end{figure}
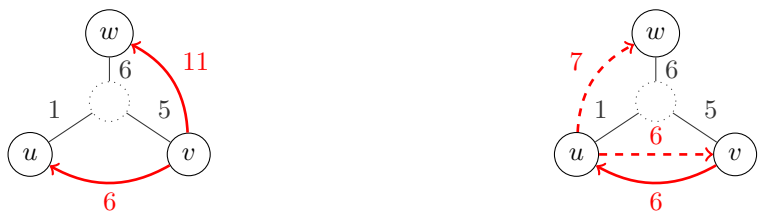

\section{Analysis}
\label{sec:analysis}
In this section we rigorously analyze \BG.
We prove the protocol's correctness with regard to Definition \ref{def:self-stabilization}
and the set of configurations given in Definition \ref{def:legal_rng}.
Furthermore, we bound the protocol's convergence time.

%\subsection{Correctness}
%\label{subsec:correctness}
The main result of this section is that \BG is indeed a self-stabilizing protocol as stated by the following theorem:
\begin{theorem}
    \label{theorem:self_stabil}
    Let $(V,d_T)$ be a tree metric.
    Then \BG is a self-stabilizing protocol that constructs an overlay with regard to $\mathcal{L}_{MST}$.
\end{theorem}

In this section, we will concentrate on initial configurations $c_0 \in C$ where $G^*_{c_0}$ is connected.
Since two nodes from different components can \emph{never} communicate with each other (cf. \cite{corona}),
the result can trivially be extended to all initial configurations.

Our proof's structure is as follows.
First, we will show that eventually the system will contain all edges of $MST(V,d_T)$ and also keeps them in all subsequent configurations.
This will be the major part of this section.
Then we show that all remaining edges that are not part of the MST but may still be part of a configuration will eventually vanish.
This proves the protocol's convergence.
Last, we prove that once the system is in a legal configuration,
the set of explicit edges does not change and no more edges that are not part of the MST are added.
This shows the protocol's closure.
Over the course of this section we will refer to all edges $e \in MST(V,d_T)$ as \textit{valid} edges.
We call all other edges \textit{invalid}.

We begin by showing that the system eventually reaches a configuration that contains all valid edges.
For the proof, we assign a potential to each configuration $c \in C$.
As the potential, we choose the weight of the minimum spanning tree that can be constructed from all implicit and explicit edges in the configuration
if we ignore their direction, i.e., we consider the MST of $G_c^*$.
Since $G_c^*$ is simply an undirected, weighted graph with unique edge weights,
it must have a unique minimum spanning tree if it is connected.
This fact is a well-known result in graph theory. %\cite{Kruskal}.
The potential is formally defined as follows:
\begin{definition}[Potential]
    \label{def:potential}
    Let $c \in C$ be a configuration and $\mathcal{M}_c := MST(E_c^*,d_T)$ the minimum spanning tree of $G_c^* := (V,E_c^*)$,
    then the potential $\Phi: C \rightarrow \mathbb{R}^+$ is defined as $\Phi(c) :=
    \begin{cases}
        \sum_{e \in \mathcal{M}_c} d_T(e)& \mbox{if $G_c^*$ is connected}\\
        \infty&\mbox{else}
    \end{cases}$
\end{definition}
%\begin{remark}
%    Note that there can be more than one minimum spanning tree per configuration.
%    But even if there are several trees, all of them must have the same weight.
%    Thus, the potential is well-defined in any case.
%\end{remark}
The weight of the globally optimal minimum spanning tree $MST(V,d_T)$
that considers all edges provides a lower bound for the potential.
Therefore, it cannot decrease indefinitely.
In the following, we show that the potential decreases monotonically and once the system reached a configuration with minimum potential
it will eventually contain all valid edges. First, we show that the potential can not increase.
\begin{lemma}
    \label{lemma:potential_decreasing}
    Consider an execution of \BG and let the system be in configuration $c \in C$.
    Further, let $c'$ be an arbitrary succeeding configuration of $c$.
    Then it holds $\Phi(c') \leq \Phi(c)$.
\end{lemma}
\begin{proof}
    To simplify notation let $E$ and $E'$ be the set of all edges in $G_c^*$ and $G_{c'}^*$ respectively.
    In the following, we will show that we can only construct equally good or better spanning trees
    from the edges in $E'$.
    Per definition, exactly one node $v \in V$ is activated in the transition from $c$ to $c'$.
    This node then executes the for-loop given in the pseudocode in Listing \ref{alg:BuildDG}.
    Let $v$ be the node that is activated and $\{v,w\} \in E$ be an edge that is delegated removed from $E$
    during its activation,
    i.e., $v$ delegates $w$ to some node $u$.
    As a result of the delegation,
    the configuration $c'$ contains the (implicit) edge $(u,w) \in E^T_{c'}$ and thus $E'$ contains the edge $\{u,w\} \in E'$.
    This allows us to view the delegation as swapping edge $\{v,w\}$ for $\{u,w\}$.

    In the following we observe the swaps $(e_1,e'_1), \dots, (e_k,e'_k)$,
    such that $e_i \in E$ is swapped for $e'_i \in E'$ in the transition from $c$ to $c'$.
    The order in which we observe these swaps must be consistent with the protocol.
    That means that two delegations
    must appear in the same order as they could in the for-loop,
    i.e., $v$ can only delegate to node whose reference's are still in its local memory.
    Next, we define $E_0, \dots E_k \subseteq V^2$
    with $E_0 := E$ and $E_i := E_{i-1} \setminus \{e_i\} \cup \{e'_i\}$ for $i>0$ as the edge sets
    resulting from these swaps.

    As the proof's main part we inductively show that each $MST(E_i, d_T)$ with $i \in \{1,\dots,k\}$
    has a lower or equal weight than $MST(E_{i-1}, d_T)$.
    For this, we distinguish between two cases.
    First, if $e_i \not\in MST(E_{i-1}, d_T)$,
    the spanning tree is
    not affected by the swap and thus the weight remains equal.
    Second, if $e_i \in MST(E_{i-1}, d_T)$, we must show
    that we can construct an equally good spanning tree in $E_i$.
    For this, consider
    $\mathcal{M}_i := MST(E_{i-1}, d_T) \setminus \{e_i\} \cup \{e'_i\}$.
    Note $\mathcal{M}_i$ and $MST(E_{i-1}, d_T)$ only differ in the edges $e_i := \{v,w\}$ and $e'_i := \{u,w\}$.
    For the delegation of $w$ to $u$ it must have held $u \prec (v,w)$ and thus $d_T(u,w) < d_T(v,w)$.
    Therefore, $\mathcal{M}_i$ has lower weight than $MST(E_{i-1}, d_T)$.
    It remains to show that $\mathcal{M}_i$ is a connected spanning tree for $V$.
    Further denote $T_v$ and $T_w$ as the subtrees of $MST(E_{i-1}, d_T)$ connected by $\{v,w\}$.
    To prove that $\mathcal{M}_i$ is a spanning tree,
    we must show that $\{u,w\}$ connects $T_v$ and $T_w$, i.e., it holds $u \in T_v$.
    Suppose for contradiction that $u \in T_w$.
    Then the path from $v$ to $u$ in $MST(E_{i-1}, d_T)$ contains the edge $\{v,w\}$.
    Further, note that $E_{i-1}$ must have contained the edge $\{v,u\}$ because $v$
    cannot delegate any node to $u$ without having a reference to $u$ itself.
    Therefore, the edges $\{v,w\}$ and $\{v,u\}$ are \emph{both} part of $E_i$ and \emph{both} connect $T_v$ and $T_w$.
    Now consider that $\{v,u\}$ is shorter than $\{v,w\}$, because a delegation requires $u \prec (v,w)$ and thus $d_T(v,u) < d_T(v,w)$.
    Hence $MST(E_{i-1}, d_T)$ could be improved by swapping $\{v,w\}$ for $\{v,u\}$.
    This is a contradiction because $MST(E_{i-1}, d_T)$ is a minimum spanning tree.
    Therefore $u \in T_v$ and the edge $\{u,w\}$ connects $T_v$ and $T_w$.

    Thus, $\mathcal{M}_i$ is a spanning tree that can be constructed solely from edges in $E_i$.
    Further, it has a lower or equal weight than $MST(E_{i-1}, d_T)$.
    The lemma then follows by a simple induction.
\end{proof}

It remains to show that the potential actually decreases until it reaches the minimum.
That means, we need to show that there cannot be a configuration with suboptimal potential where no more
delegations that decrease the potential occur.
Note that the proof of Lemma \ref{lemma:potential_decreasing}
tells us that the potential decreases if an edge $\{v,w\} \in \mathcal{M}_c$ is delegated.
Therefore, we first show that in each suboptimal spanning tree there is a node
that can potentially detect an improvement.
\begin{lemma}
\label{lemma:locally_checkable}
Let the system be in configuration $c \in C$, s.t. the potential $\Phi(c)$
is not minimum.
Then there must exist nodes $u,v,w \in V$, such that
\[
    \big( u \prec (v,w) \big) \,\, \wedge \,\, \big(\{v,u\} \in \mathcal{M}_c\big)  \,\, \wedge  \,\, \big(\{v,w\} \in \mathcal{M}_c\big)
\]
\end{lemma}
\begin{proof}
    Let $\mathcal{M}_c$ be the minimum spanning tree of a configuration $c$.
    Since the potential is suboptimal,
    there must be two nodes $v,w \in V$ with $\{v,w\} \in MST(V,d_T)\setminus\mathcal{M}_c$.
		Since $\mathcal{M}_c$ is connected, there is a path $v:=v_0, v_1, \dots, v_k:=w$
    from $v$ to $v_k$ in $\mathcal{M}_c$.

    Now consider $v_1$.
    According to Lemma \ref{lemma:equivalence} it cannot hold $v_1 \prec (v,w)$
    because $\{v,w\} \in MST(V,d_T)$.
    Thus, it holds $d_T(v,w) < d_T(v_1,w)$ or $d_T(v,w) < d_T(v_1,v)$.
    Now we distinguish between two cases:
    \begin{enumerate}
    \item Assume, it holds $d_T(v,w) < d_T(v_1,w)$.
    Next, consider that it holds $d_T(v_{k-1},w) > d_T(v_k,w)$
    because no node can be closer to $w = v_k$ than $w$ itself.
    Thus, there must be a first node $v_i$ on the path with $d_T(v_i,w) > d_T(v_{i+1},w)$.
    Since $d_T(v,w) < d_T(v_1,w)$ it further holds that $i \geq 1$.
    Therefore, the node $v_{i-1}$ is well-defined and it must hold $d_T(v_{i-1}, w) < d_T(v_{i},w)$
    because $v_{i+1}$ is the first node that is closer to $w$ than its successor.
		Hence, it holds
    $\big(d_T(v_{i-1},w) < d_T(v_{i},w)\big)$ and $\big(d_T(v_{i+1},w) < d_T(v_{i},w)\big)$
    Following Lemma \ref{thm:root} it follows that either $v_{i-1} \prec (v_{i},v_{i+1})$ or $v_{i+1} \prec (v_{i-1},v_{i})$.
    Since in both cases all of the involved edges are part of $\mathcal{M}_c$,
    the lemma follows.
    \item  Assume, it holds $d_T(v,w)<d_T(v_1,v)$. Then there must be node $v_i$, such that
  $d_T(v_i,v) > d_T(v_{i+1},v)$.
  Otherwise a simple induction from $v_0$ to $v_k$ would yield that $d_T(v,w)<d_T(v_k,v)$.
  Since $v_k = w$ this is a contradiction.
  The rest of the proof is analogous to the previous case. For the first deviator $v_i$ it holds
  $d_T(v_i,v) > d_T(v_{i+1},v)$ and $d_T(v_i,v) > d_T(v_{i-1},v)$ and thus,
  we can apply Lemma \ref{thm:root} to conclude the proof.
    \end{enumerate}
    
\end{proof}

Lemma \ref{lemma:locally_checkable} only made assumptions about edges in $G^*_c$ and did not consider the actual edges.
Since each node only has access to its local references, node $v$ can only perform a delegation if it ever has explicit
references to $u$ and $w$.
%Therefore, a node may not be able to detect and perform a delegation because the references are not or not yet in its local memory.
In the following lemma, we will see that if the potential does not decrease,
a node will eventually have the references in local memory.

\begin{lemma}
    \label{lemma:eventually_explicit}
    Let the system be in configuration $c \in C$ and let $\mathcal{M}_c$ be the minimum spanning tree of $c$.
    If the potential does not decrease, then the following two statements hold:
    \begin{enumerate}
          \item Every computation that starts in $c$ will reach a set $C_c \subset C$, such that
          \[
              \forall c^* \in C_c: \, \big( \{v,w\} \in \mathcal{M}_{c} \Rightarrow (v,w) \in E^X_{c^*} \big)
          \]
          \item \textsl{Every} succeeding configuration of $c^* \in C_c$ is in $C_c$ as well
    \end{enumerate}
\end{lemma}
\begin{proof}
    Recall from the proof of Lemma \ref{lemma:potential_decreasing}
    that if \textsl{any} edge is removed from $\mathcal{M}_c$, then the potential decreases.
     If we assume that the potential does not decrease, no edge is ever removed from $\mathcal{M}_c$.

    Now fix an edge $\{v,w\} \in \mathcal{M}_c$.
    This edge exists because $(v,w) \in E^X_c \cup E^T_c$ or $(w,v) \in E^X_c \cup E^T_c$.
    In the following we assume that $(v,w) \in E^X_c \cup E^T_c$, the other case is analogous.
    In order to prove the lemma we must show that if $(v,w) \in E^T_{c}$ it
    will become explicit and there must eventually be an explicit edge $(w,v) \in E^X_{c'}$
    in some later configuration $c' \in C$.
    Further, none these edges is ever delegated as long as the potential does not decrease.
    \begin{enumerate}
        \item If $(v,w) \in E_c^T$ is implicit in $c$, it will eventually be delivered to $v$.
        Since the potential does not decrease, $v$ never performs a delegation
        of an edge that is part of a minimum spanning tree.
        Thus, the system must reach a configuration $c'$ with $(v,w) \in E^X_{c'}$.
        By the same argument, any subsequent configuration $c' \in C$ with the same potential
        must also contain $(v,w) \in E^X_{c'}$.

        \item If $(v,w) \in E^X_c$ is explicit,
        then the edge $(w,v) \in E^T_{c'}$ will eventually be added in some later configuration $c'$ when $v$ is activated.
        This happens, because we assume, that no edge is delegated.
        If $v$ does not delegate $w$, it introduces itself upon its activation.
        Thus, it adds an implicit edge $(w,v)$ that will eventually become explicit.
    \end{enumerate}
    In conclusion:
    For each edge $\{v,w\} \in \mathcal{M}_c$ in a configuration $c$,
    there will eventually be a configuration $c'$ with edges $(v,w),(w,v) \in E^X_{c'}$.
    Since no edge of $\mathcal{M}_c$ is ever delegated as long as the potential is fixed,
    these explicit edges stay part of all subsequent configurations if the potential does not decrease.
\end{proof}
Using this fact we can finally show that the following holds:

\begin{lemma}[Convergence I]
    \label{lemma:convergenceI}
    The following two statements hold:
    \begin{enumerate}
      \item Every computation will reach a set $C_{MST} \subset C$, such that
      \[
          \forall c' \in C_{MST}: \,\,\, \big(\{v,w\} \in MST(V,d_T) \Rightarrow (v,w) \in E^X_{c'}\big)
      \]
      \item \textsl{Every} succeeding configuration of $c' \in C_{MST}$ is in $C_{MST}$ as well.
    \end{enumerate}
\end{lemma}

\begin{proof}

    The proof is structured in two parts: First, we show that the system eventually
    reaches a configuration $c_{min} \in C$ with minimum potential from any initial configuration $c \in C$.
     Second, we elaborate eventually the system will contain all valid edges once it is
     in $c_{min}$ and argue, why all subsequent configurations must also contain all
     valid edges.
    \begin{enumerate}
        \item Let $c \in C$ be an arbitrary configuration with suboptimal potential.
        Further, let $\mathcal{M}_c$ be a minimum spanning tree of that configuration.
        Following Lemma \ref{lemma:locally_checkable}, there must be nodes $u,v,w \in V$ with
        edges, such that
        \[
        u \prec (v,w) \,\, \wedge \,\, \{v,w\} \in \mathcal{M}_c \,\, \wedge \,\, \{v,u\} \in \mathcal{M}_c
        \]
        Since the potential does not decrease, we can apply Lemma \ref{lemma:eventually_explicit}.
        Thus, there will eventually be a configuration $c' \in C$ with
        \[
        (v,u) \in E^X_{c'} \,\, \wedge \,\, (v,w) \in E^X_{c'}
        \]
        This causes $v$ to eventually delegate $w$ to $u$ and decrease the potential (if the potential does not decrease otherwise before).
        Since the potential is lower bounded by the weight of $MST(V,d_T)$ that by definition cannot decrease,
        the system must eventually reach a configuration $c_{min} \in C$ with minimum potential.
        \item
        If the system is in a configuration $c_min \in C$, the potential cannot decrease further.
        Thus, we can apply Lemma \ref{lemma:eventually_explicit} and eventually the system reaches a configuration $c' \in C_{MST}$, such that:
        \[
            \{v,w\} \in MST(V,d_T) \Rightarrow (v,w) \in E^X_{c'}
        \]
        Hence, the configuration $c'$ contains all valid edges. Further, Lemma \ref{lemma:eventually_explicit}
        states that these edges are not delegated as long as the potential does not decrease.
        Since the potential is minimum, it can never decrease and therefore the statement follows.
    \end{enumerate}
    Hence, starting in any configuration $c \in C$ the system will eventually reach a configuration $c' \in C_{MST}$ with all valid edges.
    Further, all subsequent configurations of $c'$ are in $C_{MST}$ as well.
    This was to be shown.
\end{proof}

This concludes the first part of the convergence proof.
Now we know that the system eventually converges to a superset of the \DG.
It remains to show that eventually all invalid edges will vanish.

\begin{lemma}[Convergence II]
    \label{lemma:convergenceII}
    The following two statements hold:
    \begin{enumerate}
    \item Eventually each computation will reach a set of configurations $C' \subset C$, such that
    \[
        \forall c \in C': \,\,\, \big(\{v,w\} \not\in MST(V,d_T) \Rightarrow \{v,w\} \not\in E^*_{c}\big)
    \]
    \item \textsl{Every} succeeding configuration of $c' \in C'$ is in $C'$ as well.
    \end{enumerate}
\end{lemma}
\begin{proof}
    For this proof,
    we will again employ a potential function.
    The potential of a configuration $c \in C$ is the weight of the longest invalid edge.
    Formally:
    \[
      \tilde{\Phi}(c) := \begin{cases}
		\max_{e \in E^*_{c}\setminus MST(V,d_T) } d_T(e)& \textit{if  } E^*_{c}\setminus MST(V,d_T) \neq \emptyset\\      
      	0& \textit{else}\\
      \end{cases}
    \]
    If this potential is $0$, there are no invalid edges left.
    This trivially follows from the fact that all distances are greater than zero.
    Just as with the other potential, we will show that this potential (1) never increases and (2) will decrease as long as it is not minimum.
    \begin{enumerate}

        \item $\tilde{\Phi}(c)$ \textsl{cannot increase.}\\
        For the proof let $c \in C$ be an arbitrary configuration and $c' \in C$ be any succeeding configuration of $c$.
        To prove the assumption, we show that the protocol never adds an invalid edge that is longer than any existing edge.
        Let $v \in V$ be the node that is activated in the transition from $c$ to $c'$ and let
        $w \in V$ be an explicit neighbor of $v$ in $G_c$.
        Then $v$ performs one of the following two actions that add new edges to the system:
        \begin{enumerate}
            \item If $v$ introduces itself to $w$,
            it adds the implicit edge $(w,v) \in E^T_{c'}$ to the system.
            Since the edge $(v,w) \in E^X_c$ with $d_T(v,w) = d_T(w,v)$ is already present, this cannot raise the potential.
            \item If $v$ delegates $w$ to some node $u \in V$,
            then it adds the implicit edge $(u,w) \in E^T_{c'}$ to system if it was not already present before.
            Since for delegation it must hold that $d_T(u,w)<d_T(v,w)$ for the existing edge $(v,w) \in E^X_c$,
            the new edge cannot raise the potential.
        \end{enumerate}
        Thus, it holds $\tilde{\Phi}(c') \leq \tilde{\Phi}(c)$.

        \item $\tilde{\Phi}(c)$ \textsl{will eventually decrease if $\tilde{\Phi}(c)>0$.}\\
        Let $c \in C$ be an arbitrary configuration and $\{v,w\} \in E^*_{c}$ an invalid edge in $c$
        with $\tilde{\Phi}(c) = d_T(v,w)$.
        Since $\{v,w\}$ is oblivious of the true edge's direction,
        both $(v,w)$ and $(w,v)$ could be part of the configuration.
        Since the proof is analogous for both edges, we will only consider $(v,w)$
        and show that all instances of this edge will eventually be delegated.

        First, consider the case that $v$ has an explicit edge to $w$.
        Since we assume the system is in a configuration that contains all edges in $MST(V,d_T)$,
        we can use Lemma \ref{lemma:equivalence}.
        According to the Lemma, there must be a node $u \in V$ with an explicit edge $(v,u) \in E^X_c$ and $u \prec (v,w)$.
        Thus, $v$ will delegate $w$ to $u$ upon activation and add the edge $(u,w)$ with $d_T(u,w)<d_T(v,w)$.

        Second, consider the case that $(v,w) \in E^T_c$ is implicit.
        For the proof, we need to mind that there can be multiple instances of the reference to $w$ in $v$'s channel.
        The potential will only sink once all of these instances are gone.
        Therefore let $\theta_v$ be the number of references to $w$ in $v$'s channel.
        In the following, we will show that $\theta_v$ decreases to $0$.
        Note that $\theta_v$ can only be raised if some node $u \in V$ delegates $w$ to $v$ or $w$ introduces itself.
        A delegation always implies that some node $u$ has a reference to $w$ and it holds $d_T(u,w) > d_T(v,w)$.
        In that case, there exists an invalid edge $\{u,w\} \in E^*_c$, which is longer than $\{v,w\}$.
        This is impossible because $\{v,w\}$ is by assumption the longest invalid edge.
        Hence, $\theta_v$ may only increase if $w$ introduces itself.
        To do this, there must be an explicit edge $(w,v)$.
        However, we can apply the same argumentation as above for $(v,w)$
        and see that $w$ must delegate its reference of $v$ to some other node $u' \in V$
        instead of introducing itself.
        In summary, the protocol never increases $\theta_v$ and thus it can only decrease if a reference
        is delivered to $v$.
        Since this eventually happens to every reference,
        the system will reach a configuration with no references of $w$ in $v$'s channel.
    \end{enumerate}
    Hence, the potential will eventually reach $0$ and no more invalid edges are left.
    Furthermore, no more invalid edges can ever be added as this would increase the potential.
\end{proof}

Thus, we have shown that starting from any weakly connected initial configuration $c \in C$ the system will converge to a superset of the \DG and eventually to a legal configuration.
This is the combined result of Lemmas \ref{lemma:convergenceI} and \ref{lemma:convergenceII}.
To complete the proof we must show that the system once it is legal never leaves the set of legal configurations.
Formally:
\begin{lemma}[Closure]
    \label{lemma:closure}
    Let the system be in a legal configuration $c \in \mathcal{L}_{MST}$,
     then every succeeding configuration $c' \in C$ is also legal.
\end{lemma}
However, the lemma is a direct corollary of Lemmas \ref{lemma:convergenceI} and \ref{lemma:convergenceII}.
Hence, \BG is self-stabilizing with regard to Definition \ref{def:self-stabilization}.
This proves Theorem \ref{theorem:self_stabil} and concludes the analysis of the protocol's correctness.

%\subsection{Complexity}
%\label{sec:complexity}

It remains to analyze how many steps are needed until a legal configuration is reached.
Therefore, we adapt the notion of \emph{asynchronous rounds} from \cite{Dolev00}.
Each computation can be divided into rounds $R_0, \dots, R_t$ with $t \to \infty$, such that
each round $R_i$ consists of a finite sequence of consecutive configurations.
Let $c_i$ be the first configuration of $R_i$, then the rounds in the first configuration, such that:
\begin{enumerate}
    \item
    For each $v\in V$, all messages that are in $v$'s channel in configuration $c_i$
    have been delivered at any of $v$'s activations in this round.
    \item
    All nodes have been activated at least once.
\end{enumerate}
Since we assume weakly fair action execution
and
fair message receipt rounds are well-defined. Using this definition, we can show the following.
\begin{theorem}
    \label{lemma:complexity}
    \BG needs $\mathcal{O}(N^2)$ asynchronous rounds to converge to a legal configuration.
\end{theorem}

\begin{proof}
  For the proof, we show that (1) it takes at most $\mathcal{O}(N^2)$ rounds until all valid edges
  are added and (2) it takes another $\mathcal{O}(N^2)$ rounds for all invalid edges to vanish.
  \begin{enumerate}
    \item
    Consider the potential function in Definition \ref{def:potential} and note that it can lower at most $\mathcal{O}(N^2)$ times.
    In the following we show that the potential will reduce at least every $4$ rounds.
    Therefore, assume that the system is in configuration $c \in C$ which is part of round $R_i$.
    According to Lemma \ref{lemma:locally_checkable}, there are nodes $u,v,w \in V$ with $\{v,u\} \in \mathcal{M}_c$ and $\{v,w\} \in \mathcal{M}_c$
    such that the potential is reduced if $v$ delegates $w$ to $u$.
    If $v$ has the references to both these nodes in its local memory,
    it will perform the delegation upon its next activation.
    This will happen at latest in round $R_{i+1}$ because each node must be activated at least once every round.
    If $v$ does not have the references in its local memory, it will eventually receive
    them both (cf. Lemma \ref{lemma:eventually_explicit}) if the potential does not decrease.
    Assume w.l.o.g. that $w$ is the the latter of the two nodes whose reference is delivered to $v$.
    In the following, we will bound the number of rounds until the system is a configuration $c' \in C$ with $(v,w) \in E^X_{c'}$
    if potential the potential does not reduce.
    Therefore, we make the following observations:
    \begin{enumerate}
      \item If $(v,w) \in E_c^T$ is implicit in some round $R_{j}$,
      it will eventually be explicit in round $R_{j+1}$.
        \item If $(w,v) \in E^X_c$ is explicit in some round $R_j$,
        then the implicit edge $(v,w) \in E^T_{c'}$ will be added in round $R_{j+1}$
        or the potential decreases.
        Note that each node must be activated at least once every round.
        If $w$ delegates $v$, then the potential decreases because $\{v,w\} \in \mathcal{M}_c$.
        If $v$ does not delegate $w$, it introduces itself.
        Thus, $w$ adds the implicit edge $(v,w)$ in round $R_{j+1}$.
        \item If $(v,w) \in E_c^T$ is implicit in some round $R_j$, it becomes explicit in round $R_{j+1}$.
        This happens because we assume that each message will be delivered within one round.
    \end{enumerate}
    By looking at the possible combination of these cases,
    we can see that the potential reduces at latest in round $R_{i+4}$.
    Together with the fact that the potential can reduce at most $\mathcal{O}(N^2)$ times,
    our statement follows.
    \item
    Consider a configuration that already contains all valid edges and let
  the corresponding round be $R_i$. We will now show that the potential defined
    in the proof of Lemma \ref{lemma:convergenceII} reduces after at most two rounds.
    Therefore, consider the longest invalid edge $\{v,w\} \not\in MST(V,d_T)$
    that is present in first configuration of round $R_i$.
    Until the end of round $R_i$ all instances of the edge
    became explicit if they were not already.
    Thus, in the first configuration of round $R_{i+1}$, the according references
    are in the local memory of $v$ or $w$
    and there are no more implicit instances of the edge.
    The proof of Lemma \ref{lemma:convergenceII} also suggests that no further implicit
    instances are ever added. Thus, it remains to show that all explicit instances
    are removed in $R_{i+1}$.
    In the following, assume that a reference of $w$ is in $v$'s memory.
    The other case is analogous.
    According to Lemma \ref{lemma:equivalence} there is a node $u \in N_v$ with $u \prec (v,w)$.
    That means, $v$ delegates $w$ to another node
    once it is activated in round $R_{i+1}$.
    Thus, at the end of round $R_{i+1}$ all instances of $\{v,w\}$ are gone
    and the potential must decrease.
    Together with the fact that the potential can decrease at most $\mathcal{O}(N^2)$ times,
    our statement follows.
  \end{enumerate}
  Thus, after $\mathcal{O}(N^2)$ rounds the system contains all valid and no invalid edges and therefore
  is in a legal configuration.
\end{proof}
%To conclude the proof, it remains to show that all invalid edges vanish after another $\mathcal{O}(N)$ rounds.
%The proof for this can be found in the appendix. The basic idea is that an invalid edge $(v,w)$ that is removed by delegation is
%replaced by a shorter edge $(u,w)$ every round. This can only happen $N-1$ times.
%Together with the fact that there is no other mechanism that adds invalid edges, this proof the lemma.

\section{Conclusion \& Outlook}
In this work, we focused on designing and analysing self-stabilizing overlay networks that take into account the underlay.
For the tree metric we considered, it turns out that there is an extremely simple protocol for MST construction that naturally follows from some general properties of MSTs in such tree metrics (notice the close relation between Lemma~\ref{lemma:equivalence} and the protocol).
Considering different kinds of underlays (such as planar graphs or graphs with bounded growth) as well as other types of overlays than a minimum spanning tree may be possible next steps.
Of course, the high upper bound on the running time of our algorithm naturally raises the question whether a better running time can be achieved by a more sophisticated algorithm or a refined analysis.
Thus, improving on our results may also be a possible next step.
%\newpage
% We have presented and analyzed a protocol for the self-stabilizing construction of an MST
% for a tree metric.
% For future work there are two interesting directions:
% \todo{Sollte ich das schreiben. In der Einleitung schreibe ich ja, dass Bäume toll sind und jetzt sage ich, dass man was anderes will...}
% First, it would be interesting to see if an MST (or a similarly sparse topology)
% can also be constructed for underlying graphs other than trees, e.g., planar graphs or graphs with bounded growth.
% Second, one could investigate more elaborate overlay topologies
% with probably shorter paths and better resilience than an MST.

\bibliography{icalp2018-ref}

\newpage

\end{document}